\documentclass[11pt]{article}
\usepackage[T1]{fontenc}
\usepackage[utf8]{inputenc}
\usepackage[english]{babel}
\usepackage{lmodern,microtype}
\usepackage{amsmath,amssymb,amsthm,mathtools,graphicx,booktabs,float}
\usepackage[authoryear,round]{natbib}
\usepackage[hidelinks]{hyperref}
\usepackage{caption}
\usepackage[a4paper,margin=1in]{geometry}
\newtheorem{proposition}{Proposition}
\newtheorem{theorem}[proposition]{Theorem}
\newcommand{\NC}{\mathrm{NC}}
\newcommand{\B}{\mathrm{B}}

\title{Climate Cooperation and Spatial Industrial Relocation:\\ Stability, Emissions, and Regional Policy Incentives}
\author{Angela Martiradonna\\
\small Department of Economics, University of Foggia, Italy\\
\small \texttt{angela.martiradonna@unifg.it}}
\date{}

\begin{document}
\maketitle

\begin{abstract}
Climate policy changes production costs and can reshape the geography of production. We study this spatial channel in a focused model where mobile industrial activity responds to local attractiveness, agglomeration strengthens attractiveness, and carbon policy lowers the profitability of emissions-intensive locations. The model yields an explicit carbon-policy threshold separating spatial instability from stability. When emissions intensity rises with local industrial density, spatial concentration raises aggregate emissions even with unchanged total industrial activity. Embedding these dynamics in a two-jurisdiction policy game creates a wedge between regional and collective incentives: each government values the industrial activity located within its own territory, while aggregate welfare internalizes the effect of policy on economy-wide concentration. With transferable utility, cooperation can select a policy on the stable side of the threshold while the non-cooperative policy lies on the unstable side. A numerical exercise illustrates this regime change. European regional evidence, based on industrial emissions and gross value added at NUTS2 level, documents a positive emissions--activity association and identifies the empirical quantities needed to measure relocation and concentration effects. The analysis links climate cooperation to the spatial organization of production and establishes a benchmark for policy-induced industrial relocation.
\end{abstract}


\section{Introduction}

Climate policy changes the relative profitability of producing across jurisdictions, and firms respond through location, investment and output decisions. This spatial adjustment is central to carbon leakage and industrial competitiveness. When environmental policy becomes more stringent in one location, emissions-intensive activity may contract there, expand elsewhere or redirect new investment across space. The environmental consequence therefore depends on the geography created by relocation as well as on the amount of activity displaced.

The spatial configuration can itself influence subsequent adjustment. Industrial concentration creates agglomeration forces through local supply networks, labor markets, infrastructure, knowledge spillovers and market access. A policy-induced relocation may therefore alter the attractiveness of both origin and destination regions and generate further movements of production. Small initial differences can dissipate, or they can become self-reinforcing and lead to persistent industrial concentration. Climate policy can consequently affect the stability of industrial geography together with the distribution of production and emissions.

This mechanism calls for an explicitly spatial approach. A regional or aggregate measure of displaced production records the amount of relocation, while a spatial model also describes where production accumulates, how neighboring locations interact and whether the resulting configuration persists. These features become particularly relevant when emissions intensity depends on local industrial density. In that case, two allocations with the same aggregate production can generate different aggregate emissions because they differ in spatial concentration. The geography of production then becomes an economic and environmental state variable.

A large literature studies carbon leakage, environmental regulation and competitiveness. General-equilibrium and policy-design models analyze how unilateral environmental policy affects production and emissions across jurisdictions and how trade exposure, market structure and policy design shape leakage incentives
\citep{BurniauxMartins2000,Boehringer2012,BrangerQuirion2014}.
Empirical work examines trade, employment, investment and plant-location responses to environmental regulation, with particularly important effects in pollution- and energy-intensive sectors
\citep{DechezlepretreSato2017,MartinEtAl2014,NaegeleZaklan2019,DechezlepretreEtAl2022}.
Recent firm-level evidence also shows that geographically differentiated regulation can redirect investment across regions and that agglomeration economies influence destination choices
\citep{LiuEtAl2026}.

A related literature brings economic geography and agglomeration into the analysis of environmental regulation. Agglomeration can modify firms' responses to environmental stringency and shape the location of emissions-intensive production
\citep{ZengZhao2009,Feddersen2012,BenKhederZugravu2012}.
More broadly, quantitative spatial economics provides a framework for studying how policy shocks propagate across heterogeneous locations through trade, mobility and local interactions \citep{ReddingRossiHansberg2017}. Recent applications show that carbon regulation can reorganize economic activity across regions and alter regional welfare through spatial linkages
\citep{ChenEtAl2024,ConteEtAl2026,HuangEtAl2026}.

The policy component is also related to environmental federalism and strategic environmental policy with mobile firms. This literature studies how jurisdictions balance environmental objectives against the value of attracting or retaining mobile economic activity and how decentralization shapes policy incentives
\citep{DijkstraFredriksson2010,RichterRunkelSchmidt2021}.

Our analysis focuses on a further dimension that connects these strands: the stability of the geography generated by relocation. The key issue is whether a spatial allocation absorbs local disturbances or amplifies them through a feedback between industrial density and location attractiveness. This distinction matters for environmental outcomes when concentration affects emissions intensity, and it matters for policy when regional governments and a coordinated authority value relocation differently. The spatial distribution of activity therefore links carbon-policy exposure, agglomeration dynamics and strategic policy interaction.

This perspective leads to three research questions. First, under what conditions does carbon policy change the stability of the spatial distribution of industrial activity? Second, how does a change in spatial concentration affect aggregate emissions when total industrial activity is held fixed? Third, can decentralized and cooperative policy choices place the same economy in different spatial regimes?

We address these questions with a dynamic spatial model in which industrial activity responds to local profitability differences. Agglomeration, spatial spillovers and carbon-policy exposure jointly determine location attractiveness, generating feedback between the current distribution of production and subsequent relocation. We characterize the stability of the spatial allocation, connect concentration dynamics to aggregate emissions and study how decentralized and cooperative policy choices interact with the spatial regime of production.

We complement the theoretical analysis with European regional evidence. Using two measures of industrial CO$_2$ emissions together with NUTS2 industrial activity, we document the regional activity--emissions relationship that motivates the model and identify the empirical objects required for a structural analysis of relocation and concentration. The theoretical and empirical exercises together define the extensions developed in the discussion.

The remainder of the paper proceeds as follows. Section~\ref{sec:model} introduces the spatial economy. Section~\ref{sec:stability} studies spatial stability, and Section~\ref{sec:emissions} examines the emissions implications of industrial concentration. Section~\ref{sec:game} introduces regional policy choice and cooperation. Section~\ref{sec:numerical} provides a numerical illustration, Section~\ref{sec:empirical} presents the European regional evidence, and Section~\ref{sec:agenda} discusses extensions and open research questions.

\section{A one-dimensional spatial economy}
\label{sec:model}

\subsection{Industrial location and local attractiveness}

Consider a spatial economy represented by the interval
$\Omega=(0,L)$, where $x\in\Omega$ indexes location and $t\geq0$
denotes time. Industrial activity is described by a density
$u(x,t)\geq0$: the quantity $u(x,t)\,dx$ measures the amount of
mobile industrial production located in a small neighborhood of $x$
at time $t$.

The second state variable, $a(x,t)$, represents local attractiveness.
It is a reduced-form profitability signal that summarizes
location-specific factors affecting firms' incentives to operate at
$x$, including agglomeration advantages, market access,
infrastructure, and other local operating conditions. Firms respond to spatial differences
in this signal, so differences in $a$ generate incentives to relocate
production.

Industrial activity moves through two channels. The first is
undirected spatial mobility: in the absence of profitability
differences, activity tends to spread across nearby locations. The
second is directed relocation toward locations with higher
attractiveness. We represent these two forces by
\begin{equation}
\label{eq:u}
\frac{\partial u}{\partial t}
=
D_u\frac{\partial^2u}{\partial x^2}
-
\chi\frac{\partial}{\partial x}
\left[
m(u)\frac{\partial a}{\partial x}
\right],
\end{equation}
where $D_u>0$ measures undirected spatial mobility and $\chi>0$
measures the responsiveness of industrial location to spatial
differences in attractiveness.

The mobility term is
\begin{equation}
\label{eq:mobility}
m(u)
=
u\left(1-\frac{u}{K}\right),
\qquad
0\leq u\leq K,
\end{equation}
where $K>0$ represents a reduced-form local capacity for industrial
concentration. The factor $u$ implies that directed relocation cannot
originate from a location with no industrial activity. The factor
$1-u/K$ progressively limits directed mobility as local density
approaches capacity. Hence the relocation response is strongest at
intermediate levels of industrial density. Throughout the benchmark, we
restrict attention to admissible trajectories satisfying $0\leq u\leq K$;
the local stability analysis is conducted around interior homogeneous states
with $0<\bar u<K$.

Carbon policy enters the model through an effective policy intensity
$\tau(x,t)\geq0$. After normalization, $\tau$ can represent the
incidence of a carbon tax, an emissions-permit price, or other regulatory
costs associated with carbon emissions. Higher values correspond to more
stringent carbon policy. The notation $\tau(x,t)$ allows policy
intensity to vary across locations and over time. The analytical
stability benchmark below specializes it to a spatially uniform and
time-constant policy, while the two-jurisdiction game assigns
different constant policy levels to the two regions.

Let $\varepsilon(u)$ denote emissions intensity, measured as emissions
per unit of industrial activity at local density $u$. The product
$\tau(x,t)\varepsilon(u)$ therefore represents the carbon-policy burden
per unit of industrial activity at location $x$ and time $t$. The dynamic
specification allows location-specific profitability advantages to persist
and diffuse gradually instead of adjusting instantaneously.
Local attractiveness evolves according to
\begin{equation}
\label{eq:a}
\frac{\partial a}{\partial t}
=
D_a\frac{\partial^2a}{\partial x^2}
-
\beta a
+
\alpha u
-
\tau(x,t)\varepsilon(u),
\end{equation}
where $D_a,\beta,\alpha>0$. The term $\alpha u$ captures local
agglomeration benefits: a larger industrial concentration raises the
profitability of the location. The term $-\beta a$ captures the
dissipation of temporary location-specific advantages, while $D_a$
allows attractiveness conditions to spill over across neighboring
locations. The final term,
$-\tau(x,t)\varepsilon(u)$, captures the reduction in local
profitability generated by carbon-policy exposure.

Equations \eqref{eq:u} and \eqref{eq:a} generate the central spatial
feedback. A local increase in industrial density raises attractiveness
through agglomeration. A higher attractiveness level then attracts
additional industrial activity. Carbon policy acts on this feedback
by reducing the profitability of emissions-intensive production. The
model therefore allows a spatially uniform carbon instrument to change
the dynamics of industrial location by altering the strength of the
agglomeration--relocation mechanism.

\subsection{Emissions intensity and local emissions}

The benchmark specifies emissions intensity as
\begin{equation}
\label{eq:intensity}
\varepsilon(u)
=
\varepsilon_0+\varepsilon_1u,
\qquad
\varepsilon_0>0,
\quad
\varepsilon_1\geq0.
\end{equation}
The parameter $\varepsilon_0$ is baseline emissions intensity, while
$\varepsilon_1$ measures how emissions intensity changes with local
industrial density. When $\varepsilon_1>0$, denser industrial
locations are more emissions intensive at the margin. This reduced-form
relationship can capture congestion in access to cleaner energy or
infrastructure, local composition effects, or other mechanisms through
which concentration affects environmental performance. The affine
specification isolates this channel while preserving a simple
analytical structure.

It is useful to distinguish emissions intensity from local emissions.
The function $\varepsilon(u)$ measures emissions per unit of industrial
activity. Total production-based emissions per unit of space are
\begin{equation}
\label{eq:emissions_local}
E(u)
=
u\varepsilon(u)
=
\varepsilon_0u+\varepsilon_1u^2.
\end{equation}
Thus, $\varepsilon(u)$ describes the environmental intensity of
production, whereas $E(u)$ describes the emissions generated by all
industrial activity located at a given point.

This distinction determines whether spatial concentration has an
aggregate environmental effect. If $\varepsilon_1=0$, emissions are
proportional to industrial activity and relocation changes only the
spatial distribution of emissions. If $\varepsilon_1>0$, concentration
raises emissions intensity, so two spatial allocations with the same
aggregate amount of industrial activity can generate different total
emissions.

\subsection{Boundary conditions and fixed aggregate industrial activity}

The benchmark treats $\Omega=(0,L)$ as a closed spatial economy.
Industrial activity can relocate within the domain, but it cannot enter
or leave through its endpoints. To express this restriction, define
the net spatial flux of industrial activity as
\begin{equation}
\label{eq:flux}
J_u(x,t)
=
-D_u\frac{\partial u}{\partial x}
+
\chi m(u)\frac{\partial a}{\partial x}.
\end{equation}
The first component represents undirected spatial dispersion, while
the second represents directed relocation toward locations with higher
attractiveness. Equation \eqref{eq:u} can then be written as the
conservation equation
\[
\frac{\partial u}{\partial t}
=
-\frac{\partial J_u}{\partial x}.
\]

We impose no-flux conditions at the two endpoints:
\begin{equation}
\label{eq:bc}
J_u(x,t)=0,
\qquad
\frac{\partial a}{\partial x}(x,t)=0,
\qquad
x\in\{0,L\}.
\end{equation}
The condition $J_u=0$ means that no industrial activity crosses the
outer boundary of the modeled economy. Firms may relocate between
locations inside $\Omega$, but aggregate industrial activity cannot
change through inflows from or outflows to locations outside the
domain. The condition $\partial a/\partial x=0$ rules out attractiveness
gradients generated beyond the boundary and prevents the attractiveness
signal from diffusing across the endpoints.

Using \eqref{eq:flux}, the first boundary condition is equivalently
\[
D_u\frac{\partial u}{\partial x}
-
\chi m(u)\frac{\partial a}{\partial x}
=
0
\qquad
\text{at }x=0,L.
\]

The no-flux condition implies conservation of aggregate industrial
activity. Integrating \eqref{eq:u} over space gives
\[
\frac{d}{dt}
\int_0^L u(x,t)\,dx
=
-J_u(L,t)+J_u(0,t)
=
0.
\]
Hence
\begin{equation}
\label{eq:mass}
\int_0^L u(x,t)\,dx
=
U_0
\qquad
\text{for all }t,
\end{equation}
where $U_0$ is determined by the initial industrial distribution. The
average industrial density is therefore
\begin{equation}
\label{eq:mean_density}
\bar u
=
\frac{U_0}{L}.
\end{equation}

The fixed-mass assumption isolates spatial relocation from changes in
aggregate industrial scale. A decline in industrial activity in one
part of the economy must be matched by an increase elsewhere within
the domain. Entry, exit, aggregate demand responses, and changes in
total production are excluded from the benchmark. This allows the
analysis to identify separately the consequences of spatial
redistribution and concentration; these additional margins are
introduced as extensions in Section~\ref{sec:agenda}.

\section{Carbon policy and spatial stability}\label{sec:stability}

\subsection{The homogeneous allocation}

Consider a uniform and time-constant carbon policy, $\tau(x,t)=\bar\tau$. The corresponding homogeneous steady state is $(\bar u,\bar a)$, where
\begin{equation}\label{eq:abar}
\bar a
=
\frac{(\alpha-\varepsilon_1\bar\tau)\bar u-\varepsilon_0\bar\tau}{\beta}.
\end{equation}
A uniform carbon policy changes the strength of the feedback between a local increase in industrial density and the attractiveness of that location. Through this feedback, a spatially uniform instrument can affect spatial stability even when the policy level itself is uniform across locations.

To see this, consider a small spatial perturbation around the homogeneous allocation. The Neumann eigenvalues on $(0,L)$ are $\lambda_k=(k\pi/L)^2$. For each nonconstant mode $k\geq1$, the linearized dynamics are governed by
\begin{equation}\label{eq:mode_matrix}
M_k(\bar\tau)=
\begin{pmatrix}
-D_u\lambda_k & \chi m(\bar u)\lambda_k\\
\alpha-\varepsilon_1\bar\tau & -\beta-D_a\lambda_k
\end{pmatrix}.
\end{equation}
The upper-right term captures directed relocation toward a more attractive location. The lower-left term captures the effect of additional industrial density on attractiveness after accounting for carbon-policy exposure. Their product is the destabilizing feedback in the system.

\begin{proposition}[Policy threshold for spatial stability]\label{prop:threshold}
Suppose $0<\bar u<K$ and $\varepsilon_1>0$. Define
\begin{equation}\label{eq:tauc}
\tau_c
=
\frac{1}{\varepsilon_1}
\left[
\alpha-
\frac{D_u\left(\beta+D_a\pi^2/L^2\right)}{\chi m(\bar u)}
\right].
\end{equation}
If $\tau_c>0$, the homogeneous allocation is linearly unstable to spatial perturbations for $0\leq\bar\tau<\tau_c$ and linearly stable on the fixed-aggregate-activity subspace for $\bar\tau>\tau_c$. At $\bar\tau=\tau_c$, the first nonconstant mode is neutral.
\end{proposition}

\begin{proof}
For every nonconstant Neumann mode $k\geq1$, the trace of the matrix in \eqref{eq:mode_matrix} is
\[
\operatorname{tr}M_k(\bar\tau)
=-(D_u+D_a)\lambda_k-\beta<0.
\]
Its determinant is
\[
\det M_k(\bar\tau)
=
\lambda_k\left[
D_u(\beta+D_a\lambda_k)
-
\chi m(\bar u)(\alpha-\varepsilon_1\bar\tau)
\right].
\]
Since the trace is negative, mode $k$ is linearly asymptotically stable if and only if its determinant is positive. For $k\geq1$, the term in square brackets is strictly increasing in $\lambda_k$, so the first nonconstant mode, $\lambda_1=\pi^2/L^2$, is the first one that can lose stability. Setting its determinant equal to zero gives
\[
D_u\left(\beta+D_a\frac{\pi^2}{L^2}\right)
=
\chi m(\bar u)(\alpha-\varepsilon_1\bar\tau),
\]
which yields \eqref{eq:tauc}. If $\bar\tau>\tau_c$, the determinant is positive for every $k\geq1$; if $0\leq\bar\tau<\tau_c$, the determinant of the first nonconstant mode is negative and $M_1$ has one positive eigenvalue. At $\bar\tau=\tau_c$, $M_1$ has one zero eigenvalue and one negative eigenvalue. The constant mode is excluded by the fixed-aggregate-activity constraint in \eqref{eq:mass}, which proves the stated result on the admissible subspace.
\end{proof}

\subsection{Economic interpretation of the threshold}

The threshold compares two forces. The term
\[
\chi m(\bar u)(\alpha-\varepsilon_1\bar\tau)
\]
measures the strength of the agglomeration--relocation feedback. A local increase in industrial activity raises attractiveness; the attractiveness differential then attracts additional activity. The term
\[
D_u\left(\beta+D_a\pi^2/L^2\right)
\]
collects the forces that smooth spatial differences: industrial mobility, diffusion of attractiveness and dissipation of location-specific profitability advantages.

Several comparative-static implications follow directly. Stronger agglomeration $\alpha$ or greater responsiveness to attractiveness differences $\chi$ raises the policy level required for stability. Stronger undirected spatial mobility $D_u$, faster dissipation $\beta$, and stronger diffusion of attractiveness $D_a$ make the homogeneous allocation easier to stabilize. A larger spatial domain reduces the first positive eigenvalue $\pi^2/L^2$, making broad low-frequency spatial differences more persistent and increasing the policy threshold. Carbon policy therefore weakens a self-reinforcing spatial mechanism whose strength depends on both economic agglomeration and geography.

For policy, $\tau_c$ is a regime boundary. Moving $\bar\tau$ across $\tau_c$ changes whether small spatial disparities are amplified or dampened. Optimal-tax calculations address a separate welfare question. Two policy regimes that differ only modestly in the level of the carbon instrument can therefore generate qualitatively different spatial dynamics when they lie on opposite sides of the threshold.

If $\tau_c<0$, the homogeneous allocation is already locally stable at zero policy intensity, and no positive policy threshold is required for stability. If $\tau_c=0$, the first nonconstant mode is neutral at zero policy intensity and becomes stable for every $\bar{\tau}>0$.

If $\varepsilon_1=0$, a spatially uniform carbon policy changes the level of local attractiveness but not the linear agglomeration--relocation feedback. It therefore does not shift the local spatial-stability condition. The dependence of emissions intensity on local industrial density is thus the channel through which a uniform carbon instrument affects the stability threshold in the benchmark.

\section{Spatial concentration and aggregate emissions}\label{sec:emissions}

Using \eqref{eq:mass}, aggregate emissions can be decomposed as
\begin{align}
\mathcal E[u(t)]
&:=\int_0^L E(u(x,t))\,dx\nonumber\\
&=\varepsilon_0U_0
+\varepsilon_1\frac{U_0^2}{L}
+\varepsilon_1\int_0^L\big(u(x,t)-\bar u\big)^2\,dx.
\label{eq:emissions_decomposition}
\end{align}
The first two terms depend only on aggregate industrial activity and the size of the spatial domain. The final term depends on the spatial concentration of production.

\begin{proposition}[Spatial concentration and aggregate emissions]\label{prop:emissions}
For fixed aggregate industrial activity $U_0$, if $\varepsilon_1>0$ the homogeneous allocation uniquely minimizes aggregate emissions. For any admissible allocation,
\begin{equation}\label{eq:emissions_excess}
\mathcal E[u]-\mathcal E[\bar u]
=
\varepsilon_1\int_0^L\big(u(x)-\bar u\big)^2\,dx.
\end{equation}
If $\varepsilon_1=0$, aggregate emissions are independent of the spatial distribution of industrial activity.
\end{proposition}

\begin{proof}
Because $\bar u=U_0/L$,
\[
\int_0^L\big(u(x)-\bar u\big)^2\,dx
=
\int_0^L u(x)^2\,dx-\frac{U_0^2}{L}.
\]
Substituting this identity into aggregate emissions gives \eqref{eq:emissions_decomposition}. For the homogeneous allocation,
\[
\mathcal E[\bar u]
=
\varepsilon_0U_0+\varepsilon_1\frac{U_0^2}{L}.
\]
Subtracting this expression from \eqref{eq:emissions_decomposition} yields \eqref{eq:emissions_excess}. If $\varepsilon_1>0$, the right-hand side is nonnegative and equals zero only when $u(x)=\bar u$ almost everywhere. If $\varepsilon_1=0$, the concentration term vanishes and aggregate emissions equal $\varepsilon_0U_0$ for every admissible spatial allocation.
\end{proof}

The proposition gives an environmental interpretation to the stability result. When $\varepsilon_1>0$, an unstable homogeneous allocation can evolve toward a more concentrated industrial geography and, through that concentration, toward higher aggregate emissions. Conversely, a policy that stabilizes a dispersed allocation can limit the emissions increase associated with endogenous concentration. This effect operates even though the model holds total industrial activity fixed.

The distinction between scale and concentration is useful. A conventional scale effect asks how total emissions change when aggregate production changes. Here aggregate production is constant, so the only channel is spatial organization. This isolates a mechanism that is easily obscured in aggregate data: the environmental consequence of where production takes place and how densely it is concentrated.

\section{Strategic policy and cooperation}\label{sec:game}

\subsection{Regional objectives evaluated on spatial trajectories}

For the policy analysis, we normalize $L=1$ and divide the spatial economy into two jurisdictions,
\[
\Omega_1=(0,1/2),
\qquad
\Omega_2=(1/2,1).
\]
Government $i$ chooses a time-constant policy $\tau_i\in[0,\tau_{\max}]$. For every policy pair $\boldsymbol\tau=(\tau_1,\tau_2)$, the spatial system generates a trajectory $(u^{\boldsymbol\tau,s},a^{\boldsymbol\tau,s})$ from initial-state scenario $s$. The scenarios receive nonnegative weights $q_s$ with $\sum_s q_s=1$. All scenarios share the same aggregate industrial activity $U_0$ and differ only in its initial spatial distribution.

\begin{figure}[t]
    \centering
    \includegraphics[width=\textwidth]{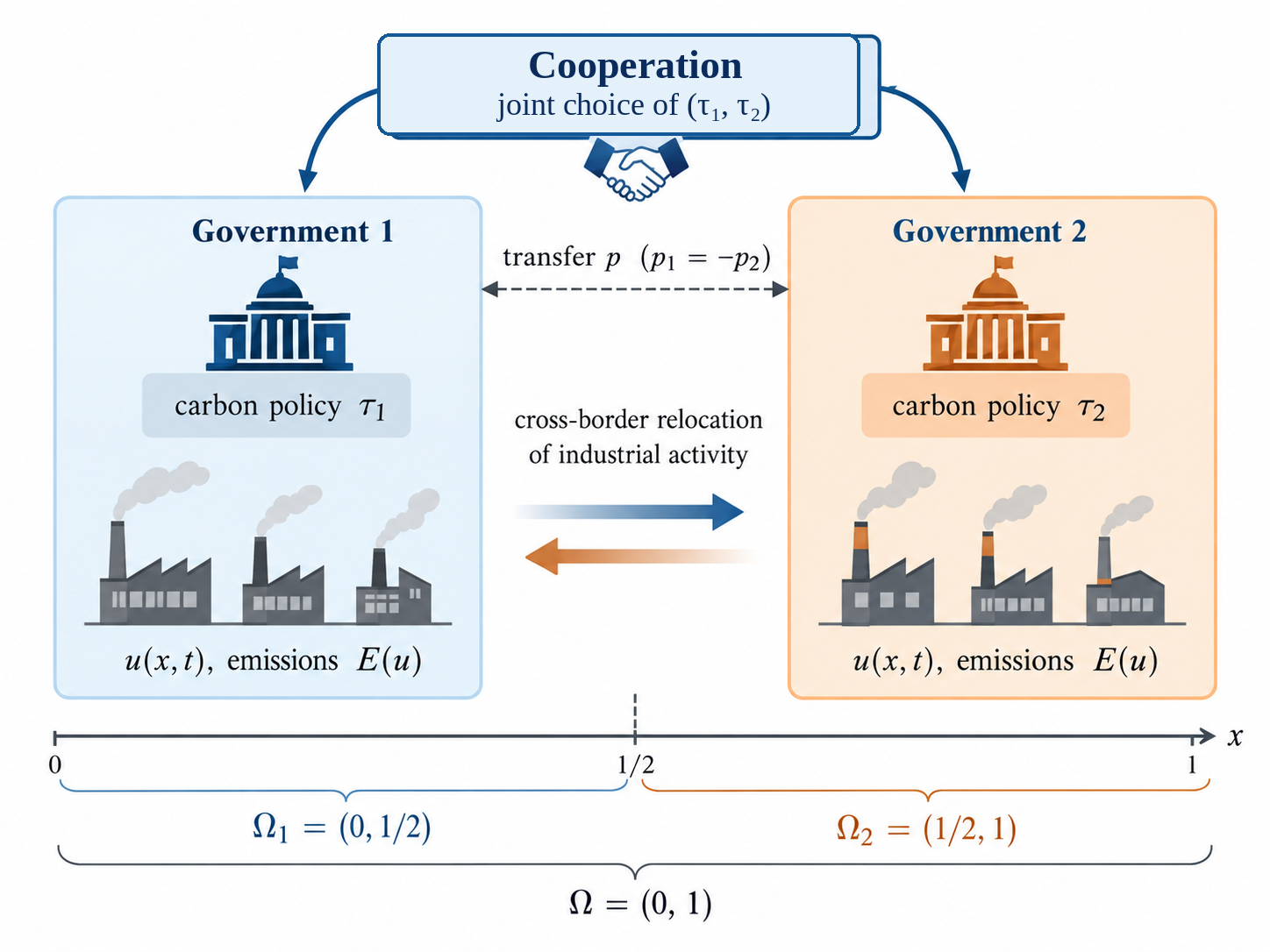}
    \caption{Spatial and institutional structure of the two-jurisdiction policy problem. The economy is divided into $\Omega_1=(0,1/2)$ and $\Omega_2=(1/2,1)$. Governments choose carbon-policy intensities $\tau_1$ and $\tau_2$, while industrial activity can relocate across the common border. Under cooperation, the two governments jointly select the policy pair and may use transfers to allocate the cooperative surplus.}
    \label{fig:bargaining_scheme}
\end{figure}

Figure~\ref{fig:bargaining_scheme} illustrates the two-jurisdiction policy environment, the cross-border mobility of industrial activity, and the cooperative policy arrangement considered below.

For each jurisdiction, define the regional spatial payoff
\begin{equation}
F_i(\boldsymbol\tau)
=
\sum_s q_s
\int_0^T e^{-\rho t}
\int_{\Omega_i}
\left[
r\,u^{\boldsymbol\tau,s}(x,t)
-
E\!\left(u^{\boldsymbol\tau,s}(x,t)\right)
\right]
dx\,dt.
\label{eq:regional_payoff}
\end{equation}
The first term values the industrial activity located within the jurisdiction, while the second records its territorial emissions cost.
Thus, relocation changes regional welfare even when aggregate industrial activity in the spatial economy is fixed.

The government's total objective is
\begin{equation}\label{eq:W}
W_i^\eta(\boldsymbol\tau)
=
b\tau_i-\frac{c}{2}\tau_i^2+h\tau_j+\eta F_i(\boldsymbol\tau),
\qquad j\neq i.
\end{equation}
The reduced-form term $b\tau_i-c\tau_i^2/2$ captures the direct domestic policy benefit and cost, $h\tau_j$ allows for an exogenous benefit from the other jurisdiction's policy, and $\eta\geq0$ controls the importance of the spatial channel. The policy variable is an effective carbon-policy index, and its numerical value measures policy intensity on a dimensionless scale.

For an interior non-cooperative policy, government $i$ satisfies
\begin{equation}\label{eq:nc_foc}
b-c\tau_i
+
\eta\frac{\partial F_i}{\partial\tau_i}(\boldsymbol\tau)
=0.
\end{equation}
The final term is the strategic spatial margin. A unilateral increase in $\tau_i$ may improve the government's direct policy objective while reducing the amount of industrial activity located in its territory. This creates a tension between environmental policy and the regional value of the industrial base.

\subsection{Why the aggregate incentive is different}

Conservation makes the distinction between regional and aggregate incentives particularly clear. Define the weighted discounted concentration measure
\begin{equation}\label{eq:VT}
\mathcal V_T(\boldsymbol\tau)
=
\sum_s q_s
\int_0^T e^{-\rho t}
\int_0^1
\left(u^{\boldsymbol\tau,s}(x,t)-\bar u\right)^2dx\,dt.
\end{equation}
Summing the two regional payoffs gives
\begin{equation}\label{eq:aggregate_identity}
F_1(\boldsymbol\tau)+F_2(\boldsymbol\tau)
=
H_\rho(T)
\left[
 rU_0-\varepsilon_0U_0-\varepsilon_1U_0^2
\right]
-
\varepsilon_1\mathcal V_T(\boldsymbol\tau),
\end{equation}
where $H_\rho(T)=\int_0^T e^{-\rho t}dt$.

Equation \eqref{eq:aggregate_identity} is the key bridge between the spatial model and the policy game. The aggregate production benefit $rU_0$ is fixed because the total industrial base is fixed. Relocating one unit of production from one jurisdiction to the other changes the distribution of regional benefits while leaving aggregate production value unchanged. At the collective level, the spatial component of policy therefore works through concentration and its emissions consequence.

This difference explains why regional and collective incentives can diverge. A government may be reluctant to tighten its own policy because it values the industrial activity that could move to the neighboring jurisdiction. At the aggregate level, industrial activity is retained after an internal relocation. The collective problem therefore asks whether the induced spatial reallocation increases or decreases economy-wide concentration and emissions.

Under cooperation with unrestricted quasilinear transfers, the governments select a policy pair that maximizes aggregate welfare whenever the resulting cooperative surplus is positive. Transfers determine how this surplus is distributed between the two jurisdictions. The policy choice is therefore governed by the aggregate objective, while transfers address the distribution of cooperative gains.

\subsection{Crossing the spatial-stability threshold}

To isolate the endogenous spatial mechanism, set $h=0$. At $\eta=0$, both regimes select the same reference policy $t_0=b/c$. In the crossing result, symmetry means that the two jurisdictions share the same primitives and that the initial-state scenarios are closed under reflection about $x=1/2$, with reflected scenarios receiving the same weight. We use the superscript $\B$ to denote the cooperative aggregate-welfare branch. Define the local responses
\begin{equation}\label{eq:slopes}
A_{\NC}(t)
=
\frac{1}{c}
\frac{\partial F_1}{\partial\tau_1}(t,t),
\qquad
A_{\B}(t)
=
\frac{1}{c}
\frac{\partial(F_1+F_2)}{\partial\tau_1}(t,t).
\end{equation}
The first response contains the regional value of retaining industrial activity. The second response removes pure redistribution across the two jurisdictions and reflects the effect of policy on economy-wide concentration.

The question is whether the difference between regional and collective marginal incentives can be large enough to move policy across the stability boundary. We study economies whose common reference policy is close to that boundary, so that a first-order institutional difference can change the spatial regime.

\begin{theorem}[Institutional crossing near the spatial threshold]\label{thm:crossing}
Suppose symmetric interior stationary policy branches exist and are continuously differentiable in $(\eta,t_0)$ in a neighborhood of $(0,\tau_c)$. Let $\tau_c$ be the stability threshold in Proposition~\ref{prop:threshold}. If there is a $\kappa>0$ such that
\begin{equation}\label{eq:crossing_condition}
A_{\NC}(\tau_c)<\kappa<A_{\B}(\tau_c),
\end{equation}
and the reference policy satisfies $b/c=\tau_c-\kappa\eta$, then for every sufficiently small $\eta>0$,
\begin{equation}\label{eq:crossing}
\tau^{\NC}(\eta)<\tau_c<\tau^{\B}(\eta).
\end{equation}
\end{theorem}

\begin{proof}
For a fixed reference policy $t_0$, differentiability of the symmetric policy branches at $\eta=0$ and the first-order conditions imply
\[
\tau^{\NC}(\eta)
=t_0+\eta A_{\NC}(t_0)+o(\eta),
\qquad
\tau^{\B}(\eta)
=t_0+\eta A_{\B}(t_0)+o(\eta).
\]
Now set $t_0=b/c=\tau_c-\kappa\eta$. Smooth dependence of the two branches on the reference policy gives
\[
\tau^{\NC}(\eta)-\tau_c
=
\eta\big[A_{\NC}(\tau_c)-\kappa\big]+o(\eta),
\]
and
\[
\tau^{\B}(\eta)-\tau_c
=
\eta\big[A_{\B}(\tau_c)-\kappa\big]+o(\eta).
\]
By \eqref{eq:crossing_condition}, the coefficient of $\eta$ in the first expression is negative and the coefficient in the second is positive. Hence, for every sufficiently small $\eta>0$,
\(
\tau^{\NC}(\eta)<\tau_c<\tau^{\B}(\eta)
\), as claimed.
\end{proof}

Condition \eqref{eq:crossing_condition} says that the regional and collective spatial margins are sufficiently different around the stability boundary. The non-cooperative government places more weight on the domestic loss of industrial activity, while the cooperative problem internalizes the benefit of limiting economy-wide concentration. The distinction is qualitative: cooperation can change the spatial regime in addition to changing the level of the carbon-policy instrument. The economic implication is that institutions can determine whether initial spatial disparities are amplified or dampened. In the unstable regime, small differences in industrial location can become persistent agglomerations. In the stable regime, the same differences decay. Cooperation can therefore affect the long-run geography of production as well as aggregate emissions.

\section{Numerical illustration}\label{sec:numerical}

A dimensionless numerical example illustrates the mechanism. We use $D_u=0.04$, $D_a=0.10$, $\chi=0.20$, $\alpha=1$, $\beta=0.40$, $K=2$, $\bar u=1$, $\varepsilon_0=0.20$ and $\varepsilon_1=1$. The policy calculation uses a conservative cell-centered finite-volume discretization on $N=40$ cells with zero boundary fluxes, and the resulting stiff system is integrated with an implicit BDF method. Welfare is averaged over two equiprobable reflected first-mode initial-state scenarios, preserving ex ante jurisdictional symmetry. Symmetric stationary candidates are obtained from the own-welfare and aggregate-welfare first-order conditions, with policy derivatives evaluated by centered differences with step $\Delta_\tau=10^{-3}$. The nonlinear trajectories reported in Table~\ref{tab:illustration} and Figure~\ref{fig:paths} are recomputed on a finer $N=160$ grid from the same first-mode industrial-density perturbation; attractiveness is initialized at the corresponding homogeneous level plus the aligned perturbation. These calculations are used to illustrate the local mechanism and do not constitute a certification of global Nash or aggregate-welfare optima. The resulting continuum threshold is
\[
\tau_c=0.445216.
\]
The value is generated by the balance in Proposition~\ref{prop:threshold}: below it, the agglomeration--relocation feedback dominates the stabilizing forces for the first spatial mode; above it, that mode decays.

With $T=24$, $\rho=0.05$, $h=0$, $b/c=0.443$, $r=2.5$ and $\eta=0.015$, the numerical policy exercise gives a symmetric non-cooperative stationary policy candidate $\hat\tau^{\NC}=0.349275$ and an aggregate-welfare stationary candidate $\hat\tau^{\B}=0.446400$. Hence
\[
\hat\tau^{\NC}<\tau_c<\hat\tau^{\B}.
\]
The numerical exercise is designed to display the mechanism associated with the theoretical crossing. The non-cooperative candidate lies clearly in the unstable region, while the aggregate-welfare candidate lies just on the stable side of the threshold.

\begin{table}[!htbp]
\centering
\caption{Illustrative policy crossing and spatial implications at $T=24$.}
\label{tab:illustration}
\small
\begin{tabular}{lccc}
\toprule
Regime & Policy & $\int_0^1(u(x,24)-\bar u)^2dx$ & $\mathcal E[u(24)]$\\
\midrule
Non-cooperative candidate & 0.349275 & 0.034267 & 1.234267\\
Aggregate-welfare candidate & 0.446400 & 0.003302 & 1.203302\\
\bottomrule
\end{tabular}
\end{table}

The concentration measure in Table~\ref{tab:illustration} is more than ten times larger under the non-cooperative candidate. Because the emissions function contains the same squared concentration term, the difference is reflected directly in aggregate emissions. The simulation therefore reproduces both links in the analytical mechanism: the institutional policy choice changes spatial stability, and the spatial regime changes emissions through concentration.

\begin{figure}[t]
\centering
\includegraphics[width=.87\linewidth]{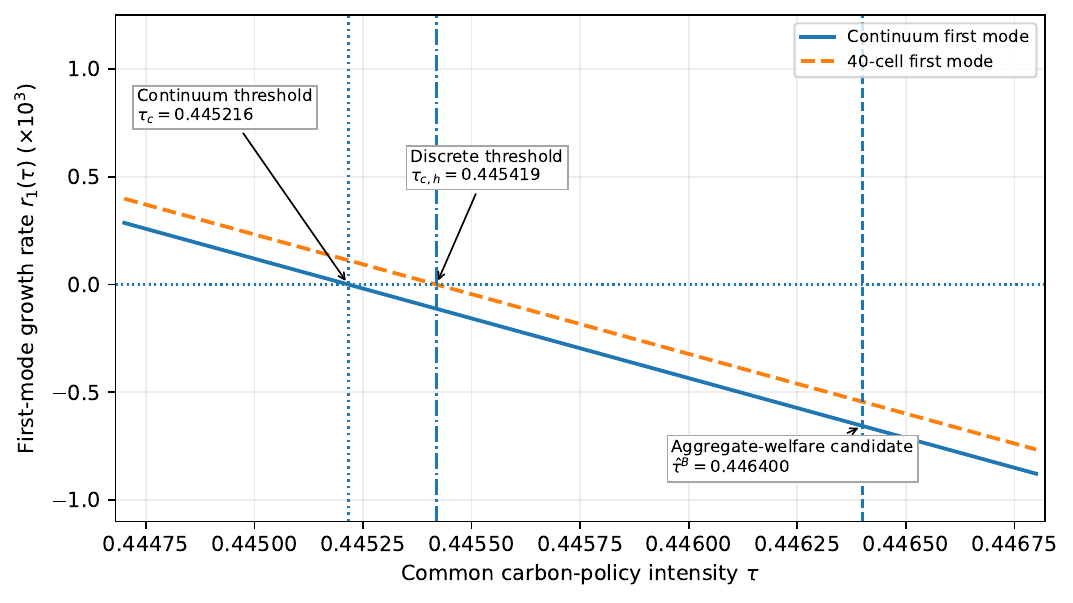}
\caption{Policy-dependent first-mode growth near the spatial stability threshold. Continuous and 40-cell discrete curves use their respective critical policy intensities. The aggregate-welfare candidate lies on the linearly stable side of the threshold.}
\label{fig:threshold}
\end{figure}

Figure~\ref{fig:threshold} resolves the narrow neighborhood of the stability boundary. The sign of the continuum first-mode growth rate changes at $\tau_c$, while the finite-grid calculation produces a nearby discrete threshold. The aggregate-welfare candidate lies just on the stable side; the non-cooperative candidate is far below this zoomed interval and, as reported above, lies in the unstable region.

\begin{figure}[H]
\centering
\includegraphics[width=.82\linewidth]{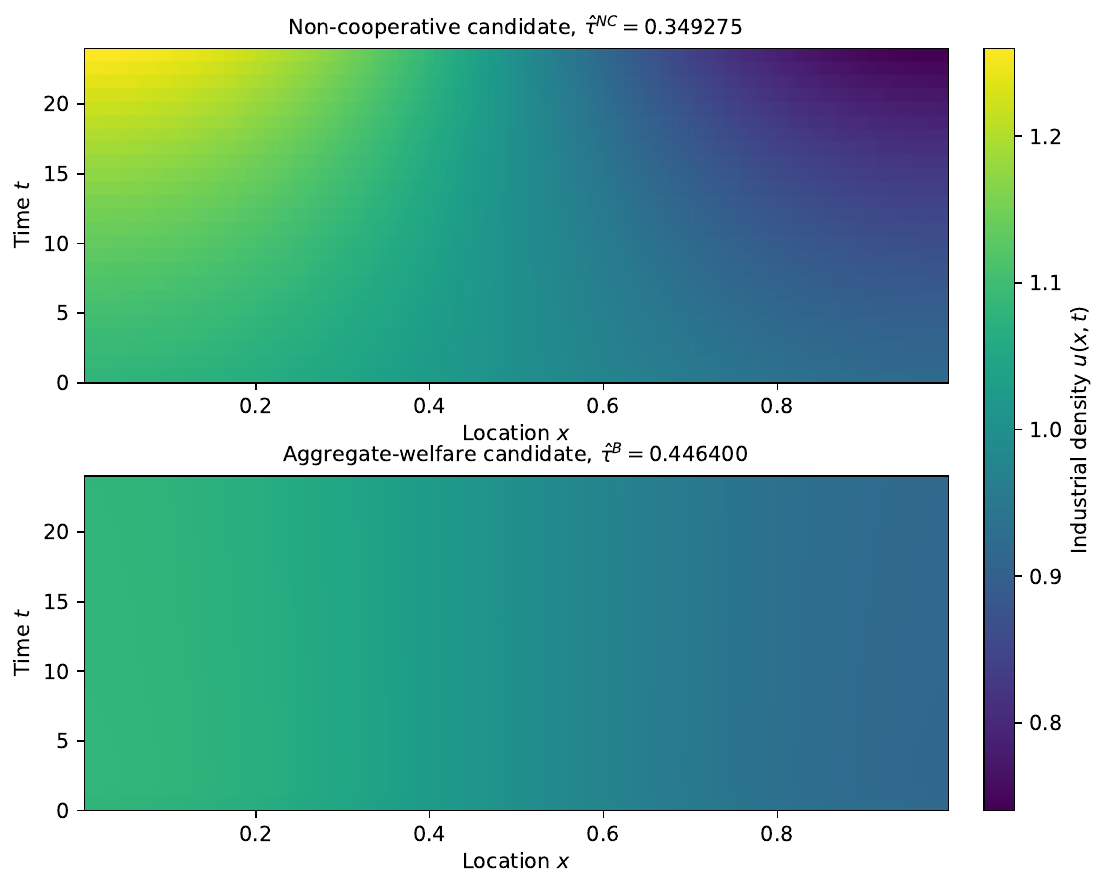}
\caption{Computed one-dimensional space--time industrial-density fields from identical first-mode perturbations under the non-cooperative and aggregate-welfare stationary policy candidates, shown on a common scale.}
\label{fig:paths}
\end{figure}

Figure~\ref{fig:paths} shows the nonlinear trajectories from identical initial perturbations. Under the non-cooperative candidate, the perturbation grows and produces a pronounced spatial concentration. Under the aggregate-welfare candidate, the spatial difference remains much smaller. The exercise illustrates how a policy difference that appears modest on the policy scale can generate a large difference in the spatial organization of production when it crosses a stability threshold.

The analysis also distinguishes endogenous agglomeration from carbon leakage. Agglomeration refers to a nonhomogeneous allocation emerging under a given policy configuration. Leakage is a policy-counterfactual concept: a tighter policy in one jurisdiction induces production or emissions to move outside that jurisdiction. The two phenomena can interact, because relocation may initiate or reinforce agglomeration elsewhere, but they are conceptually distinct and require different empirical measurements.

\section{European regional evidence}
\label{sec:empirical}

The empirical analysis provides a regional anchor for the spatial mechanism and identifies the quantities required for a structural assessment of relocation and concentration. The model points to two distinct empirical relationships: the association between industrial activity and territorial emissions, and the response of the spatial distribution of industrial activity to differences in carbon-policy exposure. The analysis below documents the first relationship using European regional data and examines its sensitivity to alternative geographical normalizations and fixed-effect structures.

\subsection{Data}

We combine four data sources. Industrial-facility CO$_2$ emissions are drawn from the European Environment Agency's industrial reporting data \citep{EEAIndustrial2026}. Regional industrial gross value added is taken from Eurostat \citep{EurostatGVA2026}, and regional area from GISCO NUTS2 geometries \citep{GISCONUTS2024}. We complement the facility-based measure with EDGAR fossil-industry CO$_2$ emissions at the subnational level \citep{EDGARNUTS2026,CrippaEtAl2024}.

The matched panel covers 2021--2023 and contains 545 region-year observations for 193 NUTS2 regions in 26 European countries. The EEA and EDGAR regressions are estimated on the same 545 region-year observations. The two emissions measures capture different aspects of regional industrial emissions. EEA data are linked to reported industrial facilities and therefore provide a close connection to regulated production sites. EDGAR provides a broader territorial inventory of fossil-industry emissions. Their joint use allows us to assess whether the regional activity--emissions relationship is consistent across two distinct measurement systems.

For the EEA measure, the estimation sample contains region-year observations with positive reported emissions from the selected industrial facilities. Non-reported cells are excluded.

\subsection{Empirical specification}

For each emissions measure, we estimate
\begin{equation}
\label{eq:empirical}
\log \mathcal E_{\ell t}
=
\gamma \log Y_{\ell t}
+\delta \log A_{\ell}
+\theta M_{\ell t}
+\mu_{c(\ell),t}
+\xi_{\ell t},
\end{equation}
where $\mathcal E_{\ell t}$ denotes industrial CO$_2$ emissions in region $\ell$ and year $t$, $Y_{\ell t}$ is industrial gross value added, $A_{\ell}$ is regional area, and $M_{\ell t}$ is the manufacturing share of industrial activity, measured as manufacturing GVA (NACE C) divided by industrial GVA (NACE B--E). The term $\mu_{c(\ell),t}$ denotes country-by-year fixed effects. Standard errors are clustered at the NUTS2 level.

The coefficient $\gamma$ measures the within-country-year association between industrial activity and emissions after accounting for regional size and broad industrial composition. We interpret it as a descriptive elasticity. Regional aggregates combine several margins---industrial scale, technology, sector composition, and spatial concentration---so structural identification of the concentration parameter $\varepsilon_1$ requires more spatially resolved measures of industrial density and emissions intensity.

\subsection{Results}

Industrial activity is positively associated with emissions under both measurement systems. The estimated coefficient on industrial gross value added is $0.435$ for EEA facility emissions and $0.674$ for EDGAR fossil-industry emissions.

\begin{table}[!htbp]
\centering
\caption{Industrial activity and regional CO$_2$ emissions in Europe, 2021--2023. Brackets report 95\% confidence intervals based on NUTS2-clustered standard errors.}
\label{tab:empirical}
\footnotesize
\setlength{\tabcolsep}{4pt}
\begin{tabular}{p{0.43\linewidth}cc}
\toprule
 & EEA facility CO$_2$ & EDGAR industry CO$_2$\\
\midrule
$\log$ industrial GVA & 0.435 & 0.674\\
95\% confidence interval & [0.195, 0.675] & [0.527, 0.822]\\
Region-year observations & 545 & 545\\
Country-by-year fixed effects & Yes & Yes\\
Area and industrial-composition controls & Yes & Yes\\
\bottomrule
\end{tabular}
\end{table}

Because both emissions and industrial gross value added enter in logarithms, the estimates imply that a 10\% higher level of industrial gross value added is associated with approximately 4.35\% higher EEA facility emissions and 6.74\% higher EDGAR fossil-industry emissions, conditional on the included controls and fixed effects. The 95\% confidence intervals exclude zero in both specifications.

The difference in magnitudes is consistent with the distinct coverage of the two emissions measures. EEA data track emissions from reported industrial facilities, while EDGAR captures a broader territorial fossil-industry inventory. The positive association is consistent across the two emissions measurement systems.

For the spatial model, the relevant implication is that regions with larger industrial bases also tend to display higher industrial emissions after conditioning on regional area, broad industrial composition, and country-specific annual factors. This empirical relationship motivates the activity--emissions margin in the theoretical analysis.

\subsection{Robustness and sensitivity}

The positive activity--emissions association persists across several alternative descriptive specifications, although its magnitude depends on geographical normalization and on the emissions accounting system. Table~\ref{tab:empirical_robustness} reports the main sensitivity checks. In the preferred area-adjusted specification, the elasticity is $0.435$ for EEA facility emissions and $0.674$ for EDGAR fossil-industry emissions. Estimating the same levels model without the regional-area control yields similar positive coefficients, $0.464$ and $0.710$, respectively.

A density-on-density specification produces larger elasticities, $0.657$ for EEA and $0.766$ for EDGAR. This specification, however, implicitly restricts the coefficients on industrial GVA and regional area in the corresponding levels model to sum to one. Region-clustered tests reject this restriction for both EEA ($p=0.0028$) and EDGAR ($p=0.0335$). We therefore retain the specification with freely estimated regional area as the main descriptive benchmark.

\begin{table}[!htbp]
\centering
\caption{Sensitivity of the industrial-activity coefficient to alternative descriptive specifications. Brackets report 95\% confidence intervals based on NUTS2-clustered standard errors.}
\label{tab:empirical_robustness}
\footnotesize
\setlength{\tabcolsep}{4pt}
\begin{tabular}{p{0.47\linewidth}cc}
\toprule
 & EEA facility CO$_2$ & EDGAR industry CO$_2$\\
\midrule
Area-adjusted baseline & 0.435 & 0.674\\
 & [0.195, 0.675] & [0.527, 0.822]\\
Levels without area control & 0.464 & 0.710\\
 & [0.234, 0.694] & [0.569, 0.851]\\
Density normalization & 0.657 & 0.766\\
 & [0.452, 0.862] & [0.650, 0.883]\\
\bottomrule
\end{tabular}
\end{table}

The two emissions measures also permit a paired same-sample comparison. Under the area-adjusted specification, the EEA coefficient is $0.239$ lower than the EDGAR coefficient, with a region-clustered 95\% confidence interval of $[-0.454,-0.024]$ ($p=0.029$). This comparison is exploratory and should be interpreted cautiously because the two sources cover different emissions-accounting perimeters and the exercise involves several specification checks. Its main implication is that the positive activity--emissions relationship is common to both sources, while the estimated magnitude is sensitive to measurement and geographical normalization.

A more demanding within-region specification further illustrates the limits of the short panel. For EEA facility emissions, a model with NUTS2 and year fixed effects yields an industrial-activity coefficient of $0.020$ with a 95\% confidence interval of $[-0.429,0.469]$. The three-year panel therefore provides no precise estimate of within-region changes in the activity--emissions relationship. This result reinforces the descriptive interpretation of the baseline estimates and motivates the use of longer panels and more spatially resolved firm or facility data for causal analysis.

\subsection{Connection with the spatial mechanism}

The regional estimates identify the first empirical relationship highlighted by the model: industrial activity and territorial emissions move together across regions. The spatial mechanism involves two additional margins. The first is the response of industrial location to relative carbon-policy exposure. The second is the effect of local industrial concentration on emissions intensity.

A structural empirical analysis would therefore combine spatially resolved firm or facility data with policy variation across locations. Firm-level location responses could provide an empirical counterpart to the relocation sensitivity $\chi$. Measures of local industrial density and emissions intensity could identify the concentration channel represented by $\varepsilon_1$. Observed changes in the spatial distribution of production would then provide an empirical counterpart to the policy-induced change in $u(x,t)$.

This distinction also separates three related empirical objects. An activity--emissions regression measures how industrial scale and emissions covary across regions. A relocation analysis measures how relative policy exposure changes the location of industrial activity. A carbon-leakage analysis adds a policy counterfactual linking tighter regulation in one jurisdiction to changes in production or emissions elsewhere. The spatial mechanism studied here connects these objects through the evolution of industrial concentration.

Europe provides a useful setting for this extension because industrial production, emissions reporting, and carbon-policy exposure can be observed across neighboring regions within a highly integrated market. Combining these data at the firm, facility, and regional levels would make it possible to separate scale, relocation, and concentration effects and to assess which spatial regime described by the model is empirically relevant.

\section{Discussion and extensions}
\label{sec:agenda}

The analysis identifies a spatial channel through which carbon policy can affect industrial organization as well as emissions. Carbon policy changes local profitability, profitability differences redirect industrial activity, and relocation changes the spatial concentration of production. When concentration feeds back into location attractiveness and emissions intensity, policy can alter the stability of the industrial allocation. The institutional environment then matters because regional and collective objectives attach different values to the relocation of mobile production.

The baseline model isolates this mechanism by holding aggregate industrial activity fixed, adopting an affine emissions-intensity function, and studying policy interaction around a local stability boundary. These assumptions yield an explicit policy threshold while also indicating the main dimensions along which the mechanism can be extended.

A central parameter in the baseline model is the response of emissions intensity to local industrial density. With
\[
\varepsilon'(u)>0,
\]
greater concentration raises aggregate emissions for a fixed industrial base and increases the carbon-policy burden associated with additional local production. Carbon pricing then weakens the agglomeration feedback and can change spatial stability.

The sign and magnitude of this relationship are empirical objects. Industrial concentration may raise emissions intensity through congestion, pressure on energy infrastructure, limited access to cleaner inputs, or changes in the composition of production. It may also improve environmental efficiency through shared infrastructure, technological spillovers, specialized services, and access to cleaner production technologies. For a general differentiable function $\varepsilon(u)$, the effect of carbon policy on the local agglomeration feedback depends on $\varepsilon'(\bar u)$. Positive, zero, and negative values therefore imply different spatial responses to climate policy.

The fixed-mass assumption provides a complementary separation. Because
\[
\int_\Omega u(x,t)\,dx=U_0,
\]
the benchmark attributes changes in aggregate emissions to the spatial organization of a given industrial base. In observed economies, carbon policy can also affect aggregate production and technology. Allowing these margins to interact with relocation would make it possible to assess how spatial, scale, and technological responses jointly shape emissions.

The threshold $\tau_c$ characterizes the response of the homogeneous allocation to small spatial perturbations. It identifies the policy level at which the local dynamics change from amplification to decay. Beyond this local boundary, nonlinear dynamics determine the spatial configuration that ultimately emerges. Initial conditions and temporary policy differences may then have persistent effects on industrial location, creating scope for multiple spatial outcomes, path dependence, and hysteresis.

The same mechanism can also be studied in richer spatial environments. Transport links, market access, input-output relationships, and network structure may alter both the propagation of relocation incentives and the spatial patterns that emerge. A network formulation provides one natural way to represent interactions across connected regions or production nodes. Related models on connected landscapes show how spatial spillovers and cooperation can be represented among interacting jurisdictions \citep{CaravaggioDiLiddoMartiradonna2025}.

The institutional result shows that regional and collective marginal incentives can select policies on opposite sides of the stability boundary,
\[
\tau^{NC}<\tau_c<\tau^B.
\]
Its economic content lies in the difference between retaining industrial activity locally and evaluating relocation from the perspective of the spatial economy as a whole. Policy coordination can therefore affect the spatial regime in addition to the level of the carbon instrument.

The current cooperative formulation isolates this distinction through transferable utility and finite-horizon policy commitment. Alternative institutional environments may change both policy selection and the distribution of cooperative gains. Asymmetry, limits on transfers, commitment, and renegotiation can all modify the incentives identified in the benchmark. Dynamic bargaining is particularly relevant when current agreements affect future bargaining opportunities and the economic state evolves between negotiations \citep{Flamini2020}.

The welfare analysis can also be extended beyond territorial emissions accounting. The present specification assigns an emissions cost where production occurs, while climate damages ultimately depend on aggregate emissions. Incorporating broader environmental damages and alternative accounting rules would allow the spatial relocation mechanism to be studied together with the global climate externality. Border adjustments, revenue recycling, and alternative transfer arrangements are further institutional dimensions that may affect the incidence of climate policy.

The regional evidence identifies the first empirical relationship highlighted by the model: industrial activity and territorial emissions are positively associated across European regions. A structural assessment of the spatial mechanism additionally requires evidence on how industrial location responds to relative policy exposure and on how local concentration affects emissions intensity.

Spatially resolved firm and facility data, combined with policy variation across locations, would help separate relocation from changes in output and technology. Such evidence would allow the spatial mechanism to be evaluated more directly and would help determine which policy-dependent spatial regimes are empirically relevant.

These extensions point to broader questions concerning industrial dynamics, institutional design, and empirical identification. The stability mechanism provides a common structure linking these directions.

\section{Conclusion}

Carbon policy can influence the geography of production through firms' location incentives. When industrial concentration reinforces local attractiveness, this spatial response can generate a stability threshold: below a critical policy level, small spatial differences are amplified; above it, they are dampened. If emissions intensity also rises with local concentration, the spatial regime has an aggregate environmental consequence even when total industrial activity is fixed.

The policy game adds an institutional dimension. A regional government values industrial activity located within its own territory, while cooperation internalizes the economy-wide effect of relocation on concentration. Regional and collective incentives can therefore select policies on opposite sides of the spatial threshold. The numerical exercise illustrates the resulting difference in spatial dynamics and aggregate emissions.

European regional evidence complements the theoretical analysis by documenting a positive association between industrial activity and industrial CO$_2$ emissions across two distinct emissions measures. The empirical results are descriptive, but they identify the margins that would be needed for a fuller structural assessment of policy-induced relocation and concentration.

The analysis therefore provides a benchmark linking climate policy, agglomeration, and institutional design through the spatial organization of industrial production. Extending this framework to richer industrial dynamics, alternative institutional settings, and more spatially resolved empirical analysis offers a natural direction for further research.

\section*{Declaration of generative AI and AI-assisted technologies}

OpenAI ChatGPT was used for language editing, editorial support, and assistance with refinement of code used for numerical analysis and figures. The author reviewed and verified all AI-assisted content and takes full responsibility for the final manuscript.

\end{document}